\documentclass[oneside]{amsart}
\usepackage{amsmath}
\usepackage{amssymb}
\usepackage[mathscr]{euscript}
\usepackage[margin=2.25cm]{geometry}
\usepackage{bbm}
\usepackage{hyperref}
\usepackage{setspace}
\usepackage{tikz}
\usepackage{color}

\newtheorem{thm}{Theorem}[section]
\newtheorem{cor}[thm]{Corollary}

\newtheorem{lem}[thm]{Lemma}
\newtheorem{lemma}[thm]{Lemma}

\newtheorem{remark}[thm]{Remark}
\newtheorem{proposition}[thm]{Proposition}

\newtheorem{corollary}[thm]{Corollary}

\newcommand{\R}{\mathbb{R}}
\newcommand{\E}{\mathbb{E}}
\newcommand{\Z}{\mathbb{Z}}

\newcommand{\s}{\sigma}

\title{Thermodynamic Limit for the Mallows Model on $S_n$}

\author{Shannon Starr}

\address{
Department of Mathematics\\
Hylan Building\\
University of Rochester\\
Rochester, NY 14627}

\email{\url{http://www.math.rochester.edu/people/faculty/sstarr/}}

\date{March 2, 2009}

\begin{document}

\onehalfspacing

\markright{}

\begin{abstract}

The Mallows model on $S_n$ is a probability distribution on permutations,
$q^{d(\pi,e)}/P_n(q)$, where $d(\pi,e)$ is the distance between $\pi$ and the identity
element,
relative to the Coxeter generators. Equivalently, it is the number
of inversions: pairs $(i,j)$ where $1\leq i<j\leq n$, but $\pi_i>\pi_j$.
Analyzing the normalization $P_n(q)$, Diaconis and Ram calculated the 
mean and variance of $d(\pi,e)$ in the Mallows model,
which suggests the appropriate $n \to \infty$ limit has 
$q_n$ scaling as $1-\beta/n$.
We calculate the distribution of the empirical measure in this limit,
$u(x,y)\, dx\, dy = \lim_{n \to \infty} \frac{1}{n} \sum_{i=1}^{n} \delta_{(i,\pi_i)}$.
Treating it as a mean-field problem, analogous to the Curie-Weiss model,
the self-consistent mean-field equations are 
$\frac{\partial^2}{\partial x \partial y} \ln u(x,y) = 2 \beta u(x,y)$, 
which is an integrable PDE, known as the hyperbolic Liouville equation.
The explicit solution also gives a new proof of formulas
for the blocking measures in the weakly asymmetric exclusion process,
and the ground state of the $\mathcal{U}_q(\mathfrak{sl}_2)$-symmetric
XXZ ferromagnet.

\vspace{8pt}
\vskip .2 cm
\noindent
{\small \bf MCS numbers: 82B05, 82B10, 60B15} 
\end{abstract}

\maketitle

\section{Introduction and Main Results}

The Coxeter generators of the symmetric group $S_n$ are the transpositions $(1,2)$, $(2,3)$,
\dots, $(n-1,n)$.
The height of a permutation is defined distance to the identity element $e$,
$$
d(\pi,e)\, =\, \min\{k \geq 0\, :\, \exists\, \tau_1,\dots,\tau_k \in \{(1,2),\dots,(n-1,n)\}\ \text{ such that }\ \pi = \tau_1 \cdots \tau_k\}\, .
$$
More generally, $d(\pi_1,\pi_2) = d(\pi_2^{-1} \pi_1,e)$.
It is easy to see that $d(\pi,e) = d(\pi^{-1},e)$.
In fact, another formula is
$$
d(\pi,e)\, =\, |\{(i,j) \in \Z^2\, :\, 1\leq i<j\leq n\ \text{ and }\ \pi_i>\pi_j\}|\, .
$$
The pairs $(i,j)$ are called inversions of $\pi$.
In \cite{DiaconisRam}, Diaconis and Ram studied the Mallows measure, which is a probability measure on $S_n$ given by
$$
\mathbbm{P}_{n}^{q}(\pi)\, =\, \frac{q^{d(\pi,e)}}{P_n(q)}\, ,
$$
with $P_n(q)$ being a normalization constant.
Actually, Diaconis and Ram studied a Markov chain on $S_n$ for which the Mallows model gives the limiting
distribution.
This was followed up by another paper on a related topic by Benjamini, Berger, Hoffman
and Mossel (BBHM) \cite{BBHM} who related the biased shuffle and the Mallows model to the asymmetric
exclusion process and the ``blocking'' measures (of Liggett, see \cite{Liggett}, Chapter VIII, especially 
Example 2.8 and the end of Section 3).
They did this using Wilson's height functions \cite{Wilson}.
We will discuss this more in Section \ref{sec:Applications}.
For now, let it suffice that Diaconis and Ram identified the explicit formula for the normalization
which they remarked is the ``Poincar\'e polynomial'':
$$
P_n(q)\, =\, \prod_{i=1}^{n} \left(\frac{q^{i}-1}{q-1}\right)\, =\, [n]_q!\, =\, [n]_q \cdots [1]_q\, ,\quad \text{where}\quad
[n]_q\, =\, \frac{q^n-1}{q-1}\, .
$$
Further references for the statistical applications\footnote{
Independently, a similar $q$-deformed combinatorial formula was explained for a problem in quantum statistical
mechanics, the ground state of the ferromagnetic $\mathcal{U}_q(\mathfrak{sl}_2)$-symmetric XXZ 
quantum spin chain, by Bolina, Contucci and Nachtergaele \cite{BCN}.
We will comment more on this in Section 8.}
of the Mallows model can be found in their paper.

Note that, physically speaking, one would define a Hamiltonian energy function $H_n : S_n \to \R$ as
$$
H_n(\pi)\, =\, \frac{1}{n-1}\, \sum_{1\leq i<j\leq n} \mathbbm{1}_{(0,\infty)}(\pi_i-\pi_j)\, .
$$
In this case, one thinks of $\pi = (\pi_1,\dots,\pi_n)$ as some type of constrained spin system,
where each of the components $\pi_1,\dots,\pi_n$ are spins in $\{1,\dots,n\}$ as in a Potts model.
The choice of the normalization of the Hamiltonian is then standard for mean-field models.
One would be most interested in the free energy
$$
f_n(\beta)\, =\, -\frac{1}{\beta n}\, \ln \sum_{\pi \in S_n} e^{-\beta H_n(\pi)}\, .
$$
For our purposes, we prefer to consider the mathematically simpler ``pressure''
$$
p_n(\beta)\, =\, \frac{1}{n}\, \ln \left(\frac{1}{n!}\, \sum_{\pi \in S_n} e^{-\beta H_n(\pi)}\right)\, .
$$
(Note that, contrary to the usual conventions of statistical physics, we have divided the partition function, which is
$Z_n(\beta) = \sum_{\pi in S_n} e^{-\beta H_n(\pi)}$, by the infinite-temperature partition function $Z_n(0)=n!$,
which is equivalent to starting with a normalized {\it a priori} measure rather than counting measure on $S_n$.)
It is trivial to see that this is given precisely by the Poincar\'e polynomial described
by Diaconis and Ram:
$$
p_n(\beta)\, =\, \frac{1}{n} \ln\frac{P_n(e^{-\beta/(n-1)})}{n!}\, =\, \frac{1}{n} \ln \frac{[n]_{e^{-\beta/(n-1)}}!}{n!}\, .
$$
With this scaling, it is also easy to calculate the limit:
$$
p(\beta)\, =\, \lim_{n  \to \infty} p_n(\beta)\, =\, \int_0^1 \ln\left(\frac{1-e^{-\beta x}}{\beta x}\right)\, dx\, ,
$$
(which can be solved explicitly using the polylogarithm function).
From this one can calculate the mean and variance in the limiting Gibbs measure.
For instance, one can calculate $\E(d(\pi,e) - n(n-1)/4)^2 \sim n^3/72$ in the uniform
measure on $S_n$.

To go beyond the statistics of $d(\pi,e)$ it seems worthwhile to study the empirical measure of $\pi$:
$$
\frac{1}{n} \sum_{i=1}^{n} \delta_{(i,\pi_i)}\, ,
$$
which is a normalized measure on $\{1,\dots,n\} \times \{1,\dots,n\}$.
More specifically, this is a random measure.
Rescaling the discrete cube $\{1,\dots,n\} \times \{1,\dots,n\}$ to $[0,1] \times [0,1]$, 
it is easy to see that the random empirical measure converges, in probability,
to the non-random Lebesgue measure, when $\beta=0$ (the uniform case).
Our main theorem generalizes this result.

\begin{thm} 
\label{thm:main}
For any $\beta \in \R$,
$$
\lim_{\epsilon \downarrow 0} \lim_{n \to \infty} 
\mathbbm{P}_n^{1-\beta/n}\left\{\left|\frac{1}{n} \sum_{i=1}^{n} f(i/n,\pi_i/n) - \int_{[0,1] \times [0,1]} f(x,y) u(x,y)\, dx\, dy\right| > \epsilon\right\}\, =\, 0\, ,
$$
for every continuous function $f : [0,1] \times [0,1] \to \R$, where
$$
u(x,y)\, =\, \frac{(\beta/2) \sinh(\beta/2)}{\big(e^{\beta/4} \cosh(\beta[x-y]/2) - e^{-\beta/4} \cosh(\beta[x+y-1]/2)\big)^2}\, .
$$
\end{thm}

Note that (one can show) the limit $\beta \to 0$ gives $1$.
The proof of Theorem \ref{thm:main} uses a rigorous version of mean-field theory,
as in the solution of the Curie-Weiss model.
An interesting feature is that the self-consistent mean-field equation leads us to the characterization of $u$ as
the solution of an integrable PDE
$$
\frac{\partial^2}{\partial x \partial y} \ln u(x,y)\, =\, 2 \beta u(x,y)\, .
$$
It is not unusual for mean-field problems to lead to integrable PDE's.
We demonstrate this briefly in the next section with the ubiquitous toy model,
the Curie-Weiss ferromagnet.

\section{Toy Model: The Curie-Weiss Ferromagnet}

We include this section merely to point out that mean-field problems often do lead to integrable PDE.
However the issue is serious: in fact there is a recent paper by Genovese and Barra which we recommend for more
details \cite{GenoveseBarra}.
Our approach merely summarizes their results (in our own words) as well as the earlier paper by Barra, himself \cite{Barra}.
The configuration space of the CW model is $\Omega_N = \{+1,-1\}^N = \{\sigma = (\sigma_1,\dots,\sigma_n)\, :\, \s_1,\dots,\s_n = \pm 1\}$.
For technical reasons, we choose the Hamiltonian as
$$
H_N(\s,t,x)\, =\, - \frac{t}{2N}\, \sum_{i,j=1}^{N} \s_i \s_j - x \sum_{i=1}^{N} \s_i\, ,
$$
we assume $t\geq 0$ and $x \in \R$.
Defining $m_N(\s) = N^{-1} \sum_{i=1}^{N} \s_i$, which takes values in $[-1,1]$, we see that
$$
H_N\, =\, - N \left(\frac{t m_N^2}{2} + x m_N\right)\, .
$$
Therefore, defining
$$
p_N(t,x)\, =\, \frac{1}{N} \ln \sum_{\s \in \Omega_N} e^{-H_N(\s,t,x)}\, ,
$$
we easily see that
$$
\frac{\partial}{\partial t} p_N(t,x)\, =\, \frac{1}{2} \langle m_N^2 \rangle_{N,t,x}\, ,
$$
and
$$
\frac{\partial^2}{\partial x^2} p_N(t,x)\, =\, N \left( \langle m_N^2 \rangle - \langle m_N\rangle^2\right)\, ,
$$
where
$$
\langle f \rangle\, =\, \langle f\rangle_{N,t,x}\, =\, \frac{\sum_{\s \in \Omega_N} f(\sigma) e^{-H_N(\s,t,x)}}{\sum_{\s \in \Omega_N} e^{-H_N(\s,t,x)}}\, .
$$
Actually it is easier to consider the ``order parameter,''
$$
u_N(t,x)\, =\, \langle m_N \rangle_{N,t,x}\, =\, \frac{\partial}{\partial x} p_N(t,x)\, ,
$$
from which $p_N(t,x)$ can be calculated by solving the ODE:
$$
\begin{cases}
\frac{\partial}{\partial x} p_N(t,x)\, =\, u_N(t,x) & \text{ for $x \in \R$,}\\
p_N(t,x) - |x| \to \frac{1}{2} t^2   & \text{ as $x \to \pm \infty$.}
\end{cases}
$$
Then we see that $u_N(t,x)$ satisfies the viscous Burgers equation (with velocity equal to the negative amplitude):
$$
\begin{cases}
\frac{\partial}{\partial t}\, u_N(t,x)\, =\, u_N(t,x)\, \frac{\partial}{\partial x}\, u_N(t,x) + \frac{1}{2N} \cdot \frac{\partial^2}{\partial x^2} u_N(t,x)
& \text{ for $t>0$ and $x \in \R$,}\\
u_N(0,x)\, =\, \tanh(x) 
& \text{ for $x \in \R$.}
\end{cases}
$$
This is an integrable PDE, using the Cole-Hopf transform.
See, for instance, Chapter 4 of Whitham, \cite{Whitham}.
Actually, this leads to a solution in terms of Gaussian integrals.
The analogous transform in spin-configuration notation is the Hubbard-Stratonovich transform:
$$
e^{N t m^2/2}\, =\, \int_{-\infty}^{\infty} \frac{e^{Nt (mx - x^2/2)}}{\sqrt{2\pi/t}}\, dx\, ,
$$
which ``linearizes'' the dependence of the Hamiltonian on $m_N$, in the exponential.
This trick is used to solve the Curie-Weiss model.
See, for example, Thompson \cite{Thompson}.

Note that in the $N \to \infty$ limit, one obtains $u(t,x) = \lim_{N \to \infty} u_N(t,x)$ being the vanishing-viscosity
solution of the inviscid Burgers equation.
Shocks correspond to phase transitions.
The Lax-Oleinik variational formula for solutions of hyperbolic conservation laws applies.
See for example, Section 3.4.2 of Evans \cite{Evans}.
In this context we claim that this is equivalent to the Gibbs variational formula, 
in the mean-field limit.
We review this next.

\section{Gibbs Variational Formula}
\label{sec:SCMFE}

Let us begin by considering a general problem in classical statistical mechanics.
Suppose that $\mathcal{X}$ is a compact metric space, and suppose that there
is a two-body interaction
$$
h : \mathcal{X} \times \mathcal{X} \to \R \cup \{+\infty\}\, .
$$
We assume that $h$ is bounded below.
Then for each $N \geq 2$, one can consider the mean-field Hamiltonian $H_N : \mathcal{X}^N \to \R \cup \{+\infty\}$
$$
H_N(x_1,\dots,x_N) = \frac{1}{N-1}\, \sum_{1\leq i<j\leq N} h(x_i,x_j)\, .
$$
Suppose that there is also an {\it a priori} measure $\mu_0$ on $\mathcal{X}$, which we assume is normalized so that
$$
\int_{\mathcal{X}} d\mu_0(x)\, =\, 1\, .
$$
Then the thermodynamic quantities are the partition function,
$$
Z_N(\beta)\, 
=\, \int_{\mathcal{X}^N} e^{-\beta H_N(x_1,\dots,x_N)}\, d\mu_0(x_1) \cdots d\mu_0(x_N)\, ,
$$
the pressure, 
$$
p_N(\beta)\, 
=\, \frac{1}{N}\, \ln Z_N(\beta)\, ,
$$
and the Boltzmann-Gibbs measure
$$
d\mu_N^{\beta}(x_1,\dots,x_N)\, 
=\, \frac{e^{-\beta H_N(x_1,\dots,x_N)}}{Z_N(\beta)}\, d\mu_0(x_1) \cdots d\mu_0(x_N)\, .
$$
Physically, it is more correct to consider the free energy rather than the pressure, $f_N(\beta) = - \frac{1}{\beta} p_N(\beta)$.
But we will consider $p_N(\beta)$, which seems slightly easier to handle, mathematically.

We will write $\mu_0^N$ for the measure $d\mu_0^N(x_1,\dots,x_N) = d\mu_0(x_1) \cdots d\mu_0(x_N)$.
Also, if $f$ is a function, then we use the short-hand $\mu(f)$ for $\int f d\mu$.
Then, according to the Gibbs variational principle, we have 
\begin{equation}
\label{eq:GibbsVariational}
p_N(\beta)\, =\, \max_{\mu_N \in \mathcal{M}_{+,1}(\mathcal{X}^N)} \frac{1}{N} \left[S_N(\mu_N, \mu_0^{\otimes N}) - \beta \mu_N(H_N)\right]\, ,
\end{equation}
where $S_N(\mu_N, \mu_0^{\otimes N})$ is the relative entropy (and $\mathcal{M}_{+,1}(\mathcal{X}^N)$ denotes all Borel probability measures
on $\mathcal{X}^N$)
$$
S_N(\mu_N, \mu_0^{\otimes N})\, =\, \begin{cases} \mu_0^N(\phi(d\mu_N/d\mu_0^N)) & \text{ if $\mu_N$ is absolutely continuous with respect to $\mu_0^N$,}\\
-\infty & \text{ otherwise,}
\end{cases}
$$
and $\phi(x) = - x \ln(x)$, which is $0$ if $x=0$.
Also, the unique $\mu_N$ maximizing the Gibbs variational formula (the ``arg-max'') is the Boltzmann-Gibbs measure $\mu_N^{\beta}$.

A natural ansatz for the optimizing measure is $\mu_N = \mu^N$, for some measure $\mu \in \mathcal{M}_{+,1}(\mathcal{X})$.
Probabilistically, this means that all the $x_1,\dots,x_N$ are independent and identically distributed.
Technically, this cannot usually be exact for finite $N$.
But it leads to a simpler formula because
$$
S_N(\mu^N,\mu_0^N)\, =\, N S_1(\mu,\mu_0)\quad \text{ and }\quad
\mu^N(H_N)\, =\, \frac{N}{2}\, \mu^2(h)\, ,
$$
and one hopes that the formula may become exact in the thermodynamic limit.
Mark Fannes, Herbert Spohn and Andre Verbeure proved that this approach is rigorous
in the $N \to \infty$ limit \cite{FSV}:

\begin{proposition}[Fannes, Spohn, Verbeure 1978]
\label{prop:FSV}
The limiting pressure exists, $p(\beta) = \lim_{N \to \infty} p_N(\beta)$,
and solves the variational problem
$$
p(\beta)\, =\, \max_{\mu \in \mathcal{M}_{+,1}(\mathcal{X})} [S_1(\mu,\mu_0) - \frac{\beta}{2} \mu^2(h)]\, .
$$
Moreover, any subsequential limit of the sequence $(\mu_N^\beta)$ is a mixture of infinite product
measures $\mu^{\infty}$, for $\mu$'s maximizing the right-hand-side of the formula above.
\end{proposition}

\begin{remark}
Note that the Gibbs variational principle (\ref{eq:GibbsVariational}) is true in general for all Hamiltonians whether they are mean-field or not.
(See, for instance, Lemma II.3.1 from Israel's monograph \cite{Israel}, or any other textbook on mathematical
statistical mechanics, for a rigorous proof which also applies directly in the thermodynamic limit.)
But the product ansatz which seems to yield the formula from the proposition is not generally valid, 
since there are nontrivial correlations in the true Boltzmann-Gibbs state.
Nevertheless Fannes, Spohn and Verbeure proved the mean-field limit in the $N\to\infty$ limit, using de Finetti's theorem
(which states that all infinitely exchangeable measures are mixtures of product states) and properties of the relative
entropy.
\end{remark}

Because one has $\mu^2 = \mu \times \mu$, one replaces the linear form $\mu_N(H_N)$
by the nonlinear one $\mu^2(h)$.
Fannes, Spohn and Verbeure actually proved their theorem more generally for quantum statistical
mechanics models, such as the Dicke maser, but it also applies to classical models.
For the quantum models, one replaces de Finetti's theorem  by the non-commutative analogue, St\"ormer's theorem.
(See \cite{Aldous} and references therein for a detailed survey of de Finetti's theorem, and refer to Fannes, Spohn and Verbeure's
paper and references therein for the noncommutative analogue, which we will not need.)
With Eugene Kritchevski, we tried to find a simpler proof of the specialization of  Proposition \ref{prop:FSV}
to the classical case.
But there were several errors in our proof, which have been brought to my attention by Alex Opaku, to whom I am grateful.
Fortunately, Fannes, Spohn and Verbeure's original paper definitely does also apply to classical models.

\section{Application to the Mallows model}

\label{sec:Mallows,}
\label{sec:Liouville}

We take for $\mathcal{X}$ the unit square $[0,1] \times [0,1]$.
Suppose that $f,g : [0,1] \to \R$ are probability densities: $f,g\geq 0$ and $\int_0^1 f(x)\, dx = \int_0^1 g(y)\, dy = 1$.
For simplicity, later on, we also assume that there are constants $0<c<C<\infty$ such that
$c \leq f,g\leq C$.
Then we take the {\it a priori} measure to be
$$
d\mu_0(x,y)\, =\, f(x) g(y)\, dx\, dy\, .
$$
We take the interaction to be 
$$
h((x_1,y_1),(x_2,y_2))\, =\, \theta(x_1-x_2) \theta(y_2-y_1) 
+ \theta(x_2-x_1) \theta(y_1-y_2)\, ,
$$
where $\theta : \R \to \R$ is the Heaviside function,
$$
\theta(x)\, =\, \begin{cases} 1 & \text{ if $x>0$,}\\
0 & \text{ if $x<0$.}
\end{cases}
$$
Since $\mu_0$ is absolutely continuous with respect to Lebesgue measure, all $x_1,\dots,x_N$ and $y_1,\dots,y_N$ are distinct,
with probability 1.
(This is why we do not bother to specify $\theta$ at the discontinuity point $0$.)

Let $X_1<\dots<X_N$ and $Y_1<\dots<Y_N$ be any points.
Then for any $\sigma, \tau \in S_N$, the symmetric group, we have
$$
\frac{d\mu_N^{\beta}}{d\mu_0^N}((X_{\sigma_1},Y_{\tau_1}),\cdots,(X_{\sigma_N},Y_{\tau_N}))\,
=\, \mathbbm{P}_N^{\exp(-\beta/(N-1))}(\sigma^{-1} \tau)\, ,
$$
where $\mathbbm{P}_N^q$ is the Mallows measure on $S_N$.
So studying the limit of the $\mu_N^{\beta}$'s gives us direct information on the limit of $\mathbbm{P}^{1-\beta/N}_N$.
For any fixed $\sigma \in S_N$, the permutation $\sigma^{-1} \tau$ is uniform on $S_N$, if $\tau$ is.
Because of this, we have the following result for the marginal of $\mu_N^{\beta}$ on $(x_1,\dots,x_N)$,
$$
\int_{\mathcal{X}^N} U(x_1,\dots,x_N)\, d\mu_N^{\beta}((x_1,y_1),\dots,(x_N,y_N))\, 
=\, \int_{[0,1]^N} U(x_1,\dots,x_N) f(x_1)\cdots f(x_N)\, dx_1\cdots dx_N\, ,
$$
and the marginal on $(y_1,\dots,y_N)$,
$$
\int_{\mathcal{X}^N} U(y_1,\dots,y_N)\, d\mu_N^{\beta}((x_1,y_1),\dots,(x_N,y_N))\, 
=\, \int_{[0,1]^N} U(y_1,\dots,y_N) g(y_1)\cdots g(y_N)\, dy_1\cdots dy_N\, ,
$$
for all bounded, continuous functions $U : (0,1)^N \to \R$.
Enforcing these conditions on the marginals,  Proposition \ref{prop:FSV} yields the following:
\begin{equation}
\label{eq:FSV}
p(\beta)\, =\, \max_{\mu \in \mathcal{M}_{+,1}(f,g)} [S_1(\mu,\mu_0) - \beta \mu^2(h)]\, ,
\end{equation}
where $\mathcal{M}_{+,1}(f,g)$ is the set of all probability measures $\mu \in \mathcal{M}_{+,1}(\mathcal{X})$
such that $d\mu(x,y)$ has marginals $f(x) dx$ and $g(y) dy$.

Suppose that $\mu$ is any arg-max of the right-hand-side of (\ref{eq:FSV}).
Since we have chosen $\mu_0$ to be absolutely continuous with respect to Lebesgue measure on $\mathcal{X}$,
the same must be true of $\mu$. Otherwise the relative entropy would be $-\infty$.
So we can write
$$
d\mu(x,y)\, =\, u(x,y)\, dx\, dy\, .
$$
Then it is easy to see that the Euler-Lagrange equations for (\ref{eq:FSV}) are
\begin{equation}
\label{eq:EulerLagrange}
\ln u(x,y)\, =\, \ln f(x) + \ln g(y) + C - \beta \int_{\mathcal{X}} u(x',y') [\theta(x-x') \theta(y'-y) + \theta(x'-x) \theta(y-y')]\, dx' dy'\, ,
\end{equation}
for some constant, $C<\infty$.
Therefore, $u$ solves the equation
\begin{equation}
\label{eq:EL}
\begin{cases}
u(x,y)\, =\, \frac{1}{\mathcal{Z}} f(x) g(y) e^{- \beta \int_{\mathcal{X}} h((x,y),(x',y')) u(x',y')\, dx'\, dy'}
& \text{ for $(x,y) \in \mathcal{X}$,}\\
\int_0^1 u(x,y)\, dy\, =\, f(x) & \text{ for $x \in [0,1]$,}\\
\int_0^1 u(x,y)\, dx\, =\, g(y) & \text{ for $y \in [0,1]$,}
\end{cases}
\end{equation}
where $\mathcal{Z}$ is a normalization constant.

Since $u(x,y)$ solves an integral equation it can be differentiated both with respect to $x$ and $y$.
Doing so yields the partial differential equation
\begin{equation}
\label{eq:LiouvillePDE}
\frac{\partial^2}{\partial x\partial y} \ln u(x,y)\, =\, 2 \beta u(x,y)\, ,
\end{equation}
known as the hyperbolic Liouville equation.
This equation arises naturally in differential geometry,
related to the problem of choosing a metric on a given manifold.
I am very grateful to S.G.~Rajeev for important information regarding this PDE.
One of the facts he imparted is the symmetry of the differential equation
under the following general transformation:
\begin{equation}
\label{eq:symmetry}
\begin{split}
v(x,y)\, =\, F'(x) G'(y) u(F(x),G(y)) \qquad&\\
 \Rightarrow
\frac{\partial^2}{\partial x \partial y} \ln v(x,y)\, 
&=\, \frac{\partial^2}{\partial x \partial y} u(F(x),G(y)) + \frac{\partial}{\partial y}\left(\frac{F''(x)}{F'(x)}\right) 
+ \frac{\partial}{\partial x} \left(\frac{G''(y)}{G'(y)}\right)\\
&=\,  \frac{\partial^2}{\partial x \partial y} u(F(x),G(y))\\
&=\, 2 \beta F'(x) G'(y)  u(F(x),G(y))\\
&=\, 2 \beta v(x,y)\, .
\end{split}
\end{equation}
So, if $\frac{\partial^2}{\partial x\partial y} \ln u = \beta u$ then the same is true for $v(x,y) = F'(x) G'(y) u(F(x),G(y))$.

Our real goal is to solve the Euler-Lagrange equation (\ref{eq:EL}).
But as a first step,
we want to consider the Cauchy problem for (\ref{eq:LiouvillePDE}).
In other words, we want to consider the problem
\begin{equation}
\label{eq:Liouville}
\begin{cases}
\frac{\partial^2}{\partial x \partial y} \ln u(x,y)\, =\, 2 \beta u(x,y)
& \text{ for $(x,y) \in [0,L_1] \times [0,L_2]$,}\\
u(x,0)\, =\, \phi(x) & \text{ for $x \in [0,L_1]$,}\\
u(0,y)\, =\, \psi(y) & \text{ for $y \in [0,L_2]$,}
\end{cases}
\end{equation}
for some $L_1,L_2 > 0$ and $\phi : [0,L_1] \to \R$, $\psi : [0,L_2] \to \R$ both positive and continuous.

Note that $\frac{\partial^2}{\partial x \partial y}$ is a wave operator, with characteristics directed along $x$ and $y$.
Specifically, defining $\xi = (x+y)/\sqrt{2}$ and $\zeta = (x-y)/\sqrt{2}$, we have
$\frac{\partial^2}{\partial x \partial y} = \frac{1}{2} (\frac{\partial^2}{\partial \xi^2} - \frac{\partial^2}{\partial \zeta^2})$,
the usual wave operator.
Therefore, D'Alembert's formula for solutions of the wave equation allow us to reformulate (\ref{eq:Liouville})
as an integral equation, 
\begin{equation}
\label{eq:CauchyIntegral}
\ln u(x,y)\, =\, \ln \phi(x) + \ln \psi(y) - \ln \alpha + 2 \beta \int_{[0,x] \times [0,y]} u(x',y')\, dx'\, dy'\, ,
\end{equation}
which we prefer.
This equation is supposed to be solved for all $(x,y)  \in [0,L_1] \times [0,L_2]$.
We have introduced the number $\alpha = \phi(0)$, which we also assumed
equals $\psi(0)$, for consistency since both are supposed to give $u(0,0)$.
(Note that the initial surface, $([0,L_1] \times \{0\}) \cup (\{0\} \times [0,L_2])$, is {\it not}
a non-characteristic surface. This is the reason that our Cauchy problem
does not require initial data for the tangential derivative of $u$ even though the wave equation
is second order.)
We refer to Evans textbook for PDE's, (especially Section 2.4 on the wave equation and Section 4.6 on the Cauchy-Kovalevskaya theorem).

As we will see, the symmetry (\ref{eq:symmetry}) is the key to solving both the Euler-Lagrange equation (\ref{eq:EL}) and 
the Cauchy problem (\ref{eq:CauchyIntegral}).

\section{The Cauchy Problem}

We start with uniqueness for the Cauchy problem.

\begin{lemma}
\label{lem:IVP}
For any $L_1,L_2>0$, the Cauchy problem (\ref{eq:CauchyIntegral}) having $\phi=\psi=\alpha=1$
has at most one solution in the class of nonnegative integrable functions.
\end{lemma}

\begin{proof}
Since $\phi=\psi=\alpha=1$, equation (\ref{eq:CauchyIntegral}) simplifies to
$$
\ln u(x,y)\, =\, 2 \beta \int_{[0,x] \times [0,y]}  u(x',y')\, dx'\, dy'\, .
$$
Assuming that $u$ is nonnegative and integrable, this implies that $\ln u$ is bounded and continuous.
Then, using these properties in the right-hand-side of the equation again (similarly as one does to prove elliptic regularity)
we deduce that $\ln u$ is continuously differentiable and globally Lipschitz.
In particular, it is continuous up to the boundary.

Now suppose that there are two solutions $u$ and $v$.
Letting $z = \ln u - \ln v$, we have
$$
z(x,y)\, =\, 2 \beta \int_0^x \int_0^y [1-e^{-z(x',y')}] u(x',y')\, dx'\, dy'\, .
$$
Since both $\ln u$ and $\ln v$ are bounded, we see that $z$ is as well.
Therefore, there exists a constant $K<\infty$ such that $|1 - e^{-z}| \leq K |z|$
for all values of $z$ in the range.
So we have
$$
|z(x,y)|\, \leq\, \beta K \|u\|_{\infty} \int_0^x \int_0^y |z(x',y')|\, dx'\, dy'\, .
$$
A version of Gronwall's lemma then implies that $z \equiv 0$.
We outline this now, although our argument can probably be improved.

Let $Z(t) = \sup\{ |z(x,y)|\, :\, (x,y) \in (0,L_1) \times (0,L_2)\, ,\ xy\leq t\}$.
Then we obtain, after making the change of variables $(x,y) \mapsto (x,t)$ where $t= xy$, and 
using Fubini-Tonelli to integrate over $x$ first,
$$
Z(t)\, \leq\,\beta K \|u\|_{\infty} \int_0^t \ln(t/t') Z(t')\, dt'\, .
$$
We rewrite this as 
$$
Z(t)\, \leq\, \beta K \|u\|_{\infty} \int_0^t [\ln(t/L_1L_2) - \ln(t'/L_1L_2)] Z(t')\, dt'\, .
$$
Since $\ln(t/L_1L_2)\leq 0$ for $t\leq L_1L_2$, and since $Z\geq 0$, we can drop
the term $\ln(t/L_1L_2) Z(t')$ in the integrand to obtain
$$
Z(t)\, \leq\, \beta K \|u\|_{\infty} \int_0^t |\ln(t'/L_1L_2)| Z(t')\, dt'\, .
$$
Finally, setting $\zeta(t) = \int_0^t |\ln(t'/L_1L_2)| Z(t')\, dt'$, this leads to
$$
\zeta'(t)\, \leq\, \beta K \|u\|_{\infty} |\ln(t/L_1L_2)| \zeta(t)\, .
$$
By Gronwall's inequality (see for example Appendix B of Evans \cite{Evans}),
we obtain
$$
\zeta(t)\, =\, e^{\beta K \|u\|_{\infty} \int_0^t |\ln(t'/L_1L_2)|\, dt'} \zeta(0)\, =\, e^{\beta K \|u\|_{\infty} (1+|\ln(t/L_1L_2)|) t/L_1L_2} \zeta(0)\, .
$$
But $\zeta(0) = 0$. Hence $\zeta(t)=0$ for all $t$.
This implies $Z(t)=0$ for all $t$ which implies $z(x,y) = 0$ for all $x,y$.
\end{proof}

Next we derive the explicit solution of (\ref{eq:CauchyIntegral}), for the case $\phi=\psi=\alpha=1$.

\begin{corollary}
\label{cor:Cauchy}
Suppose that $L_1,L_2>0$ and either $\beta\leq 0$ or $L_1L_2 < 1/\beta$.
Then the unique solution of the Cauchy problem (\ref{eq:CauchyIntegral}) with $\phi=\psi=\alpha=1$
is
$$
u(x,y)\, =\, (1 - \beta xy)^{-2}\, .
$$
\end{corollary}

\begin{proof}
Uniqueness was proved in Lemma \ref{lem:IVP}, and it is trivial to check that this solves the PDE (\ref{eq:LiouvillePDE}).
Therefore, assuming that $u$ is integrable on $[0,L_1] \times [0,L_2]$, we may derive D'Alembert's formula
by standard calculus:
\begin{align*}
\ln u(x,y)\, 
&=\, \int_0^x \frac{\partial}{\partial x} \ln u(x',y)\, dx' + \ln \psi(y)\\
&=\, \int_{(0,x) \times (0,y)} \frac{\partial^2}{\partial x \partial y} \ln u(x',y')\, dx'\, dy'
+ \int_0^x \frac{\partial}{\partial x} \ln \phi(x')\, dx' + \ln \psi(y)\\
&=\, 2 \beta \int_{(0,x) \times (0,y)} u(x',y')\, dx'\, dy' + \ln \psi(y) + \ln \phi(x) - \ln \alpha\, .
\end{align*}
The only issue is to check integrability, which amounts to checking $\inf_{(x,y) \in [0,L_1] \times [0,L_2]} 1-\beta xy>0$.
This holds if and only if 
$\beta\leq 0$ or $L_1 L_2 < 1/\beta$.
\end{proof}

Let us briefly explain one approach to deriving this formula.
For nonlinear PDE's one always first guesses a scaling solution, in hopes of finding an explicit formula.
Because of the hyperbolic nature it makes sense to look for a solution $u(x,y) = U(xy)$
for some $U(z)$.
This leads to the ODE
$$
\frac{d}{dz} \ln U(z) + z \frac{d^2}{dz^2} \ln U(z)\, =\, 2 \beta U(z)\, ,
$$
which can also be expressed as
$$
\frac{d}{dz} \left(z\, \frac{d}{dz} \ln U(z)\right)\, =\, 2 \beta U(z)\, .
$$
The idea of using a power law solution is natural because the derivative of the logarithm results in a power law, itself.
Trying $U(z) = (1+cz)^p$ leads to 
$$
\ln U(z)\, =\, p \ln(1+cz)\quad \Rightarrow\quad
z \frac{d}{dz} \ln \phi(z)\, =\, \frac{cpz}{1+cz}\, =\, p - \frac{p}{1+cz}\quad 
\Rightarrow \quad 
\frac{d}{dz} \left(z\, \frac{d}{dz} \ln U(z)\right)\, =\, \frac{cp}{(1+cz)^2}\, .
$$
So, taking $p=-2$ and $c=-\beta$, this solves the equation, and gives
$U(z) = (1-\beta z)^{-2} \Rightarrow u(x,y) = (1-\beta x y)^{-2}$.
Finally, we are led to the solution of the general Cauchy problem.

\begin{corollary}
\label{cor:GenCauchy}
Suppose that $\phi, \psi : [0,1] \to \R$ are continuous and satisfy
$c\leq \phi,\psi\leq C$, for some constants $0<c<C<1$.
Also suppose that $\phi(0) = \psi(0) = \alpha$ for some $\alpha$.
Then the Cauchy problem (\ref{eq:CauchyIntegral})
has a solution if and only if $\beta \leq 0$ or $\int_0^{1} \phi(x)\, dx \int_0^{1} \psi(y)\, dy < \alpha/\beta$.
In case a solution exists, it is unique and equals
\begin{equation}
\label{eq:GenSol}
u(x,y)\, =\, \frac{\alpha \phi(x) \psi(y)}{(\alpha-\beta \Phi(x) \Psi(y))^{2}}\, ,
\end{equation}
where $\Phi(x) = \int_0^x \phi(x')\, dx'$ and $\Psi(y) = \int_0^y \psi(y')\, dy'$.
\end{corollary}

\begin{proof}
Suppose that $u$ is any solution of (\ref{eq:CauchyIntegral}).
Let $v$ be given by
$$
v(x,y)\, =\, \frac{u(\Phi^{-1}(\alpha^{1/2} x),\Psi^{-1}(\alpha^{1/2} y))}{\alpha \Phi'(\Phi^{-1}(\alpha^{1/2} x))\Psi'(\Psi^{-1}(\alpha^{1/2} y))}\, .
$$
Then, by the symmetry (\ref{eq:symmetry}), $v$  is a solution of Liouville's PDE (\ref{eq:LiouvillePDE}) on
the domain $[0,\alpha^{-1/2} \Phi(1)]\times [0,\alpha^{-1/2} \Psi(1)]$.
But $v(x,0) = v(0,y) = 1$ because $\Phi' = \phi$ and $\Psi'=\psi$.
So uniqueness, the conditions for existence, and the formula for the solution all follow from Lemma \ref{lem:IVP} and Corollary \ref{cor:Cauchy}.
\end{proof}

\section{Solving the Euler-Lagrange Equation}

By general principles, we know that a solution of (\ref{eq:EL}) always exists: specifically, the optimizer in Proposition \ref{prop:FSV}.
Next we calculate it, and prove uniqueness.

\begin{lem}
\label{lem:marginal}
If $f=g=1$, then the unique solution of (\ref{eq:EL}) 
is given by (\ref{eq:GenSol}) for 
$$
\phi(z)\, =\, \psi(z)\, =\, \frac{\beta e^{-\beta z}}{1-e^{-\beta}}\, ,\quad
\Phi(z)\, =\, \Psi(z)\, =\, \frac{1-e^{-\beta z}}{1-e^{-\beta}}\quad \text{and}\quad
\alpha\, =\, \frac{\beta}{1-e^{-\beta}}\, .
$$
\end{lem}

\begin{proof}
Suppose $u$ solves (\ref{eq:EL}). 
Note that $\lim_{x \to 0} h((x,y),(x',y')) = \theta(y-y')$ for all $(x',y') \in \mathcal{X}$.
By the dominated convergence theorem, this implies
$$
\lim_{x \to 0} \int_{\mathcal{X}} h((x,y),(x',y')) u(x',y')\, dx'\, dy'\, 
=\, \int_0^1 \int_0^1 \theta(y-y') u(x',y')\, dx'\, dy'\,
=\, \int_0^1 \theta(y-y')\, dy'\, =\, y\, ,
$$
where we used the fact that $\int_0^1 u(x',y')\, dx' = 1$ for all $y'$.
So
$$
\psi(y)\, =\, \lim_{x \to 0} u(x,y)\, =\, \frac{1}{\mathcal{Z}} e^{-\beta y}\, .
$$
Similar arguments lead to $\phi(x)\, =\, \lim_{y \to 0} u(x,y)\, =\, \frac{1}{\mathcal{Z}} e^{-\beta x}$.
Since $\int_0^1 u(x,y)\, dy = 1$ for all $x \in [0,1]$, it again follows from the dominated
converge theorem, taking the limit $x \to 0$, that $\int_0^1 \psi(y)\, dy$ must also be $1$.
So $\mathcal{Z} = (1-e^{-\beta})/\beta$.
Checking, the reader will easily see that this gives the stated value for $\phi$, $\psi$ and $\alpha$.
Integrating, it also leads to $\Phi$ and $\Psi$.

Uniqueness follows from uniqueness of the Cauchy problem, Corollary \ref{cor:GenCauchy}.
Since this is the only possible solution, and since a solution exists, this must be it.
\end{proof}

Substituting in, and simplifying leads to the formula
\begin{equation}
\label{eq:solution}
u(x,y)\, =\, \frac{(\beta/2) \sinh(\beta/2)}{\big(e^{\beta/4} \cosh(\beta[x-y]/2) - e^{-\beta/4} \cosh(\beta[x+y-1]/2)\big)^2}\, .
\end{equation}
Therefore, we arrive at the final formula.

\begin{cor}
As long as $c\leq f,g\leq C$ for some $0<c<C<1$, and $\int_0^1 f(x)\, dx = \int_0^1 g(y)\, dy = 1$,
the unique solution of (\ref{eq:EL}) is
$$
u(x,y)\, =\, \frac{(\beta/2) \sinh(\beta/2) f(x) g(y)}{\big(e^{\beta/4} \cosh(\beta[F(x)-G(y)]/2) - e^{-\beta/4} \cosh(\beta[F(x)+G(y)-1]/2)\big)^2}\, ,
$$
where $F(x) = \int_0^x f(x')\, dx'$ and $G(y) = \int_0^y g(y')\, dy'$.
\end{cor}

\begin{proof}
Suppose that $u$ is a solution of (\ref{eq:EL}) under the conditions stated. 
Define 
$$
v(x,y)\, =\, \frac{u(F^{-1}(x),G^{-1}(y))}{f(F^{-1}(x)) g(G^{-1}(y))}\, ,
$$
analogously to the proof of Corollary \ref{cor:GenCauchy}.
Note that $F$ and $G$ are continuously, strictly increasing bijections of $[0,1]$.
Using (\ref{eq:EL}), we see that
\begin{align*}
\ln v(x,y)\, 
&=\, \ln u(F^{-1}(x),G^{-1}(y)) - \ln f(F^{-1}(x)) - \ln g(G^{-1}(y))\\
&=\, -\ln \mathcal{Z} - \beta \int_{\mathcal{X}} h((F^{-1}(x),G^{-1}(y)),(x',y')) u(x',y')\, dx'\, dy'\, .
\end{align*}
Making the change-of-variables $x'' = F(x')$ and $y'' = F(y')$, we see that $dx' = dx''/f(F^{-1}(x''))$
and $dy' = dy''/g(G^{-1}(y''))$.
So we have
$$
\ln v(x,y)\, =\, - \ln \mathcal{Z} - \beta \int_{\mathcal{X}} h((F^{-1}(x),G^{-1}(y)),(F^{-1}(x''),G^{-1}(y''))) v(x'',y'')\, dx''\, dy''\, .
$$
But the Heaviside function satisfies $\theta(F(x)-F(x')) = \theta(x-x')$ for any continuous, strictly increasing function $F$.
For this reason,
$$
h((F^{-1}(x),G^{-1}(y)),(F^{-1}(x''),G^{-1}(y'')))\,
=\, h((x,y),(x'',y''))\, .
$$
In other words, $v$ also solves (\ref{eq:EL}), except that
$$
\int_0^1 v(x,y)\, dy\, =\, \int_0^1 v(x,y)\, dx\, =\, 1\, ,
$$
using the change-of-variables formula, again.
So uniqueness and the formula follows from Lemma \ref{lem:marginal}.
\end{proof}

\section{Proof of Main Result}

We now explain the minor details needed to go from Proposition \ref{prop:FSV} to a proof of Theorem \ref{thm:main}.
According to Fannes, Spohn and Verbeure's result, $\mu_{N}^{\beta}$ must converge weakly to a 
mixture of i.i.d., product measures, each of whose 1-particle marginal optimizes $S_1(\mu,\mu_0) - \frac{\beta}{2} \mu^2(h)$.
But $\mu_N^{\beta}$ has marginals on $(x_1,\dots,x_N)$ and $(y_1,\dots,y_N)$ equal to the product measures of $f(x)\, dx$
and $g(y)\, dy$, respectively.
Therefore, according to the weak law of large numbers (WLLN), we know that all the $\mu$'s in the support of the directing
measure for the limit of $\mu_N^{\beta}$, must have $x$ marginal equal to $f(x)\, dx$
and $y$ maginal $g(y)\, dy$.
Hence, this constraint can be imposed when looking for an optimizer.
This is actually a relevant comment because all optimizers, for all choices of {\it a priori} measure $\mu_0$, have the same
value/pressure: that due to the Mallows measure on $S_n$.
For concreteness, we will now take $f=g=1$.

Now suppose that $\mu$ optimizes the Gibbs formula.
It must be absolutely continuous with respect to $\mu_0$ in order to not have the relative entropy equal to $-\infty$.
So we can write 
$$
d\mu(x,y)\, =\, u(x,y)\, dx\, dy\, ,
$$
where $u(x,y)$ is absolutely continuous.
Choosing any continuous function $\phi : [0,1] \times [0,1] \to \R$, 
with
$$
\int_{\mathcal{X}} u(x,y) \phi(x,y)\, dx\, dy\, =\, 0\, ,
$$
we can take
$$
u_{\epsilon}(x,y)\, =\, (1+\epsilon \phi(x,y)) u(x,y)\, .
$$
For $|\epsilon| < 1/\|\phi\|_{\infty}$, we have that $u_{\epsilon}$ is a probability measure.
It is easy to see that
$$
S_1(\mu_{\epsilon},\mu_0)\, =\, S_1(\mu,\mu_0) - \epsilon \int \phi u \ln u
- \int [1+\epsilon \phi] \ln[1+\epsilon \phi] u\, .
$$
Since $1+\epsilon \phi$ is bounded away from $0$ (and infinity) for the $\epsilon$ we are considering,
it is clear that both integrals above are well-defined.
Moreover, it is clear that
$$
\int [1+\epsilon \phi] \ln[1+\epsilon \phi] u\, =\, o(\epsilon)\, ,
$$
because $\ln[1+\epsilon \phi] = \epsilon \phi + o(\epsilon)$ and $\int \phi u = 0$.
A similar calculation also shows that
$$
\mu_{\epsilon}^2(h)\, =\, \mu^2(h) + \epsilon \mu^2([\phi(x,y)+\phi(x',y')] h) + O(\epsilon^2)\, .
$$
Since $\mu$ is supposed to be the optimizer, the terms linear in $\epsilon$ must vanish:
$$
\int \phi u \ln u\, =\, \frac{\beta}{2}\, \mu^2([\phi(x,y) + \phi(x',y')] h)\, .
$$
Since $h$ is symmetric, by varying over all $\phi$ orthogonal to $u$, we deduce that
$$
u(x,y) \ln u(x,y)\, =\, \beta u(x,y) \int_{\mathcal{X}} h((x,y),(x',y')) u(x,y)\, dx'\, dy' + C u(x,y)\, ,
$$
for some constant $C$.
(The reason we cannot assume $C=0$ is because we left out one direction for $\phi$, namely
the direction parallel to $u$, so that there is an indeterminacy in this direction, as seen using the Riesz representation
theorem.)
In other words, we have just deduced equation (\ref{eq:EulerLagrange}).
On the other hand, we have also proved that this equation has a unique solution given by (\ref{eq:solution}).
Therefore, $\mu_N^{\beta}$ does converge weakly to the i.i.d., product measure of $\mu$, where $d\mu(x,y) = u(x,y)\, dx\, dy$.

Because of all this, if we take the empirical measure with respect to $\mu_N^{\beta}$,
$$
\frac{1}{N}\, \sum_{i=1}^{N} \delta_{(x_i,y_i)}\, ,
$$
then this does satisfy just the type of convergence claimed
in Theorem \ref{thm:main}.
But, taking the order statistics $X_1<\dots<X_N$ and $Y_1<\dots<Y_N$, we do have $(x_i,y_i) = (X_{\sigma_i},Y_{\tau_i})$
for some permutations $\sigma,\tau \in S_N$.
Moreover (by commutativity of addition)
$$
\frac{1}{N}\, \sum_{i=1}^{N} f(x_i,y_i)\, =\, 
\frac{1}{N}\, \sum_{i=1}^{N} f(X_i,Y_{\pi_i})\, ,
$$
where $\pi = \tau \sigma^{-1}$.
As noted before, $(X_1,\dots,X_N)$ and $(Y_1,\dots,Y_N)$ are distributed as the order statistics coming from Lebesgue
measure, the effect of the Hamiltonian is only present in the Mallow model $\mathbbm{P}_N^{\exp(-\beta/(N-1))}$-measure of $\pi$.
By the WLLN for the order statistics, we see that, defining
$$
g_N(x,y)\, =\, \sum_{i,j=1}^{N} f(X_i,Y_j) \mathbbm{1}_{((i-1)/N,i/N]}(x) \mathbbm{1}_{((j-1)/N,j/N]}(y)\, ,
$$
we have that the random function $g_N$ converges in probability to $f$, everywhere in $(0,1]\times (0,1]$.
Therefore, since
$$
\frac{1}{N}\, \sum_{i=1}^{N} f(x_i,y_i)\, =\, \frac{1}{N}\, \sum_{i=1}^{N} g_N(i/N,\pi_i/N)\, ,
$$
we do deduce the theorem from the corresponding result for $\mu_N^{\beta}$.
Finally note that taking $\exp(-\beta/(n-1))$ versus $1-\beta/n$ in the theorem does not matter,
since the probability measures are continuous with respect to $\beta$, and $\exp(-\beta/(n-1)) = 1 - \beta(1+o(1))/n$.

\section{Applications}
\label{sec:Applications}

The ground state of the $\mathcal{U}_q(\mathfrak{sl}_2)$-symmetric XXZ quantum spin system, and the invariant measures
of the asymmetric exclusion process on an interval can be obtained from $\mathbbm{P}_N^{q}$.
See Koma and Nachtergaele's paper \cite{KomaNachtergaele} and Gottstein and Werner's paper \cite{GottsteinWerner} 
for information about the XXZ model.
For information about the blocking measures and the asymmetric exclusion process, we find it convenient to refer to Benjamini, Berger, Hoffman and Mossel (BBHM), \cite{BBHM}.
The reader can easily deduce information for the XXZ model, since there is a perfect dictionary between
these two.
An excellent reference for this is Caputo's review \cite{Caputo}.

An interesting perspective on the ground state of the quantum XXZ ferromagnet was discovered by Bolina, Contucci and Nachtergaele
in \cite{BCN}.
They viewed the ground state of the quantum spin system as a thermal Boltzmann-Gibbs state for a classical
model at inverse temperature $\beta = \ln(q^{-2})$.
The state space they considered was the set of all up-right paths from $(0,0)$ to $(m,n) \in \mathbb{Z}^2$ (with $m,n\geq 0$).
The Hamiltonian energy function for such a path is the energy under the path, and above the $x$-axis.
Note that the Hamiltonian for the Mallows model also has a graphical representations as the number of ``crossings'' 
of the permutation.
Using their representation, they explained some symmetries of the ground state of the XXZ model,
and obtained estimates which were later useful in their follow-up paper, \cite{BCN2}.
The two models are related, but only the Mallows model is manifestly a mean-field model.

We consider the (nearest neighbor) asymmetric exclusion process on $\{1,\dots,N\}$, with hopping rate to the left $p$
and hopping rate to the right $1-p$, and $q = (1-p)/p$.
We no longer use $p$ or $p_N$ for the pressure, instead we use it for the hopping rate as expressed above.
As BBHM explain, the invariant measure of the ASEP is a push-forward of $\mathbbm{P}_N^q$.
Given a permutation $\pi \in S_N$ and a particle configuration $\eta = (\eta_1,\dots,\eta_N) \in \{0,1\}^N$, let 
$\pi \eta = (\eta_{\pi_1},\dots,\eta_{\pi_N})$.
Let $(1^k,0^{N-k}) = (1,\dots,1,0,\dots,0)$ with $k$ $1$'s and $N-k$ $0$'s.
Then, taking a random permutation $\pi$, distributed according to $\mathbbm{P}_N^q$, and letting
$$
\eta^{(k,N-k)} = \pi (1^k,0^{N-k})\, ,
$$
the law of $\eta^{(k,N-k)}$ is the invariant measure for the ASEP, with $k$ particles and $N-k$ holes.
As BBHM explain, this is an instance of Wilson's general height function approach to tiling and shuffling \cite{Wilson}\footnote{Because 
of this, let us note that the ground state of the XXZ model is also a 
projection, or marginal, of the Mallows model for permutations (using the correspondence between
the ASEP and the XXZ model \cite{Caputo}).
This raises an interesting point for further consideration: are other integrable models
projections of mean-field models?}.

The question we can answer is the non-random limiting density of $\eta^{(k,N-k)}$ in the scaling limit, $N \to \infty$, $p_N = \frac{1}{2} + \beta/4N$,
$k_N = \lfloor{y N}\rfloor$.
(Note that this corresponds to $q_N = 1- \beta(1+o(1))/N$.)
Namely, for a continuous function $f : [0,1] \to \R$, we have
$$
\lim_{\epsilon \downarrow 0} \lim_{N \to \infty} \mathbbm{P}^{1-\beta/N}_{N}\left\{
\left| \frac{1}{N}\, \sum_{i=1}^N f(i/N) \eta^{(\lfloor{y N}\rfloor,\lceil{(1-y)N}\rceil)}_i - \int_0^1 f(x) \rho(x;y)\, dx\right| > \epsilon\right\}\, =\, 0\, ,
$$
for all $y \in [0,1]$, where
$$
\rho(x;y)\, =\, \int_0^y u(x,y')\, dy'\, .
$$
The scaling $p_N = \frac{1}{2}  + \beta/4N$ is the regime typically called ``weakly asymmetric.''
See, for example, Enaud and Derrida's paper \cite{EnaudDerrida}, following the matrix method used, for example by Derrida,
Lebowitz and Speer \cite{DerridaLebowitzSpeer}.
Note that while they considered the nonequilibrium case, we consider the particle conserving, equilibrium case.
On the other hand, we are sure that the formula above is known.

The integral for $\rho(x;y)$ is readily evaluated.
Setting $\phi$, $\psi$, $\Phi$, $\Psi$ and $\alpha$ as in Lemma \ref{lem:marginal},
\begin{align*}
\rho(x;y)\, 
&=\, \int_0^{y} \frac{\alpha \phi(x) \psi(y')}{(\alpha  - \beta \Phi(x) \Psi(y'))^2}\, dy'\\
&=\, \frac{\alpha \phi(x)}{\beta \Phi(x) (\alpha - \beta \Phi(x) \Psi(y'))} \bigg|_0^y\\
&=\, \frac{\phi(x) \Psi(y)}{\alpha - \beta \Phi(x) \Psi(y)}\, .
\end{align*}
Substituting in, and doing minor algebraic simplifications, we obtain
$$
\rho(x;y)\, =\, \frac{(1-e^{-\beta y}) e^{-\beta x}}{(1-e^{-\beta}) - (1-e^{-\beta x})(1-e^{-\beta y})}\, .
$$
From this formula it is obvious that the $\beta \to 0$ limit recovers $\rho(x;y) \equiv y$, as it should (for the symmetric case).
Also, after further ``simplifications,'' we obtain
$$
\rho(x;y)\, =\, \frac{e^{\beta(\frac{1}{2}-x)/2} \sinh(\beta y/2)}{e^{\beta/4} \cosh(\beta [x-y]/2) - e^{-\beta/4} \cosh(\beta[x+y-1]/2)}\, .
$$
In particular, one can observe that the particle-hole/reflection symmetry is manifest in this formula due to the invariance under the transformation 
$(\beta,x) \mapsto (-\beta,1-x)$.

Finally, we note that we can partially undo the scaling limit by taking $\beta \to \infty$ with $x=y+t/\beta$ (assuming $0<y<1$).
Approximating $\sinh(\beta y/2) \approx \frac{1}{2} e^{\beta y/2}$ and noting that $e^{-\beta/2} \cosh(\beta [x+y-1]/2) \to 0$ since
$|x+y-1|<1$, we obtain
$$
\rho(x;y) \to \frac{1}{1+e^t}\, .
$$
This is not correctly normalized due to the fact that $dx = dt/\beta$, and $\beta \to \infty$.
On the other hand, this does recover the actual lattice scaling limit for the density (modulo a reflection),
as has been previously calculated for the XXZ model by Dijkgraaf, Orlando and Reffert in Appendix A of \cite{DijkgraafOrlandoReffert}.

\section*{Acknowledgements}

This research was supported in part by a U.S.\ National Science Foundation
grant, DMS-0706927.
I am very grateful to the following people for useful discussions and suggestions:
S.~G.~Rajeev, Alex Opaku, Carl Mueller, Bruno Nachtergaele, Wolfgang Spitzer and Pierluigi Contucci.
I also thank the anonymous referees for their useful suggestions for improvement.

\end{document}